\documentclass[conference]{IEEEtran}%

\usepackage{amsmath,amsfonts}
\usepackage[T1]{fontenc}
\usepackage{algorithmic}
\usepackage{setspace}
\usepackage{algorithm}
\usepackage{array}
\usepackage{textcomp}
\usepackage{stfloats}
\usepackage{verbatim}
\usepackage{graphicx}
\usepackage{cite}
\usepackage[bottom]{footmisc}
\usepackage{balance}
\usepackage{amsthm}
\usepackage{mathtools}
\usepackage[utf8]{inputenc}
\usepackage{float}
\usepackage{ragged2e}
\usepackage[table,xcdraw]{xcolor}
\usepackage{multirow}
\usepackage{amssymb}

\newtheorem{lemma}{Lemma}
\newtheorem{remark}{Remark}
\newtheorem{corollary}{Corollary}
\newcommand\blfootnote[1]{%
  \begingroup
  \renewcommand\thefootnote{}\footnote{#1}%
  \addtocounter{footnote}{-1}%
  \endgroup
}

\ifCLASSINFOpdf
\else
\fi

\begin{document}

\title{SLIPT-Enabled Ground-to-UAV FSO Systems with Optical Reconfigurable Intelligent Surfaces}
\author{
\IEEEauthorblockN{Luna Sugiyama, Phuc V. Trinh, and Shinya~Sugiura$^*$}
\IEEEauthorblockA{Institute of Industrial Science, The University of Tokyo, Tokyo, Japan\\
E-mail: \{lunagracesugiyama@g.ecc.u-tokyo.ac.jp, trinh@iis.u-tokyo.ac.jp, sugiura@iis.u-tokyo.ac.jp\} 
}
}

\markboth{}
{Shell \MakeLowercase{\textit{et al.}}: Bare Demo of IEEEtran.cls for Journals}
\maketitle
\begin{abstract}
This paper proposes an optical reconfigurable intelligent surface (ORIS)-assisted ground-to-unmanned aerial vehicle (UAV) free-space optical (FSO) communication system empowered by simultaneous lightwave information and power transfer (SLIPT). To overcome the line-of-sight (LoS) limitation of FSO-based SLIPT systems, we introduce an ORIS that reflects the laser beam towards a non-LoS UAV receiver. We model and analyze the combined channel characteristics, incorporating atmospheric loss, turbulence-induced fading, pointing error, and angle-of-arrival fluctuations due to UAV hovering. We derive closed-form expressions for harvested energy, outage probability, and symbol error rate (SER). Numerical results show that integrating ORIS improves EH efficiency while maintaining manageable outage and SER performance. 
\end{abstract}
\begin{IEEEkeywords}
Free-space optical (FSO) communication, simultaneous lightwave information and power transfer (SLIPT), optical reconfigurable intelligent surface (ORIS).
\end{IEEEkeywords}

\IEEEpeerreviewmaketitle
\section{Introduction}
\blfootnote{Preprint for publication in \textit{IEEE International Conference on Communications (ICC)}, Glasgow, Scotland, UK, May 2026, DOI: 10.1109/ICC59461.2026.11586808. $\copyright$ 2026 IEEE. Personal use of this material is permitted. Permission from IEEE must be obtained for all other uses, in any current or future media, including reprinting/republishing this material for advertising or promotional purposes, creating new collective works, for resale or redistribution to servers or lists, or reuse of any copyrighted component of this work in other works.}

The evolution from fifth generation to sixth generation (6G) wireless networks' capabilities, such as global coverage, ultra-high data rate transmission, ultra-low latency, low power consumption, and high energy efficiency~\cite{RefWorks:RefID:33-2023road}. Although conventional radio-frequency (RF) technology has been the cornerstone of wireless communications, it is increasingly constrained by spectrum scarcity and significant energy consumption. In this context, 
Optical wireless communication encompassing laser-based free space optical (FSO) communications and light emitting diode (LED)-based visible light communications (VLC), has emerged as a compelling alternative to the conventional radio-frequency (RF) technology~\cite{RefWorks:RefID:34-2024recent}.
Moreover, non-terrestrial networks in sixth generation (6G) introduce new applications, such as unmanned aerial vehicles (UAVs) and satellite communications, where conventional RF communications may be limited or infeasible.
Concurrently, the battery-limited Internet of Things (IoT) devices necessitate innovative solutions for sustained operation, where energy harvesting (EH) is paramount.
This dual need for high-speed data and continuous power has catalyzed interest in simultaneous lightwave information and power transfer (SLIPT), a paradigm where a single optical signal is leveraged for both communications and energy delivery~\cite{RefWorks:RefID:3-2018simultaneous}. 

The core principle of SLIPT is to utilize a light source to concurrently transmit information and deliver wireless power to a receiving device, thereby prolonging the operational lifetime of energy-constrained nodes~\cite{RefWorks:RefID:35-2024sustainabilitydriven}.
This technology is particularly transformative for mobile applications, such as drones acting as aerial base stations, without the need to land, significantly enhancing their persistence and utility~\cite{RefWorks:RefID:22-2024simultaneous}.
At the receiver (Rx), the incoming optical signal has to be processed for two distinct purposes: information decoding (ID) and EH.
Several architectures and signal processing techniques have been proposed, including time switching (TS) and power splitting (PS) schemes~\cite{RefWorks:RefID:6-2024slipt}.
While LED-based VLC is suitable for indoor SLIPT applications~\cite{RefWorks:RefID:36-2024slipt}, 
laser-based FSO systems are far more effective for long-range, outdoor scenarios~\cite{RefWorks:RefID:37-2023energy}.

Despite its advantages, the primary limitation of laser-based FSO is its strict requirement for a clear line-of-sight (LoS) path between the transmitter (Tx) and Rx.
For applications such as drone-based networks operating at low altitudes, maintaining a constant LoS link is typically challenging, which severely undermines system reliability~\cite{RefWorks:RefID:38-20203d}.
A traditional approach to mitigate this is the use of relay nodes to forward the signal around obstacles~\cite{RefWorks:RefID:21-2008relayassisted}.
However, conventional relays, which require their own laser source, power supply, and complicated pointing and tracking hardware, introduce substantial cost, complexity, and power consumption.
A more efficient solution has recently been proposed in the form of optical reconfigurable intelligent surfaces (ORIS).
An ORIS is a planar surface composed of a large number of passive reflecting elements, such as mirror arrays or metasurfaces, which can dynamically alter the phase and amplitude of an incident light beam.
By intelligently coordinating these reflections, an ORIS can effectively steer the optical beam towards a desired location, facilitating communications in a non-line-of-sight (NLoS) environment, for both terrestrial~\cite{RefWorks:RefID:7-2021intelligent} and non-terrestrial~\cite{RefWorks:RefID:43-2025optical} scenarios.

The integration of ORIS into SLIPT systems presents a promising avenue for overcoming the LoS limitation while retaining the benefits of optical power and data transfer.
Prior research has begun to explore this synergy, but the focus has been limited.
For instance, ORIS-assisted SLIPT has been considered for indoor systems using resonant beams~\cite{RefWorks:RefID:29-2024reconfigurable} and for short-range applications using LEDs~\cite{RefWorks:RefID:30-2025optical}.
These studies have established the foundational concept but have not addressed the unique challenges and opportunities of outdoor, long-range systems.
There has been research on laser-based SLIPT; however, this system can only be implemented when LoS is available and suffers from EH loss due to beam broadening over long distances~\cite{RefWorks:RefID:6-2024slipt}.
To the best of our knowledge, the design and performance analysis of an ORIS-assisted, laser-based SLIPT system for outdoor NLoS applications has not been investigated in the literature.

Against the above background, the novel contributions of this paper are as follows.
We propose an ORIS-assisted SLIPT architecture using FSO systems for ground-to-UAV scenarios. By applying linear and quadratic phase-shift (QPS) profiles, the ORIS can effectively control the beam width to improve EH while managing outage and error performance. Pointing errors and angle-of-arrival (AoA) fluctuations are modeled using generalized, independent Gaussian distributions with non-zero means, which may differ in their statistical parameters. We derive novel closed-form expressions for the end-to-end statistical channel, including the probability density function (PDF) and cumulative distribution function (CDF), while accounting for atmospheric attenuation, turbulence, pointing errors, and AoA fluctuations. To assess system performance with M-ary phase shift keying (MPSK), we present closed-form solutions for metrics such as harvested energy, outage probability, and symbol error rate (SER). 

\section{System Model}
\label{model}
\subsection{System Configurations}
As depicted in Fig.~\ref{fig:system_model}, we investigate a SLIPT-enabled ground-to-UAV FSO system, assisted by an ORIS.
We assume that there is no direct link between the ground station and the UAV due to obstructions, such as buildings or trees.
Hence, an ORIS is strategically deployed to establish two LoS links with the Tx and Rx, effectively overcoming the Tx-Rx LoS blockage. For simplicity, we consider the ORIS to be positioned midway between the ground station (Tx) located on a building rooftop and the UAV (Rx), with all components at the same altitude.
In 6G networks, UAVs are expected to operate as aerial base stations with FSO backhaul links~\cite{RefWorks:RefID:37-2023energy},~\cite{RefWorks:RefID:38-20203d} and millimeter-wave (mmWave) links for mobile ground users~\cite{RefWorks:RefID:41-2021lineofsight}, as illustrated in Fig.~\ref{fig:system_model}.
As battery-powered drones have limited flight time, EH becomes essential to ensure the uninterrupted service.
This work addresses the critical challenge of maintaining both ID and EH over the ORIS-assisted FSO backhaul link using SLIPT to support sustained UAV operation.

The ORIS reflects the optical signal from the ground station and directs it toward the UAV.
This ORIS is of the metasurface type, composed of LC molecules, which act as passive elements designed to manipulate the properties of the incident beam.
Two primary phase profiles are employed to manipulate the reflected beam: the linear phase shift (LPS) profile and the QPS profile.
The LPS profile facilitates the generalized Snell's law of reflection and redirection of the beam from the Tx to the Rx.
In contrast, the QPS profile focuses the optical beam at a distance from the ORIS, reducing the beam width of the reflected beam by applying a phase-shift profile that changes quadratically~\cite{RefWorks:RefID:43-2025optical}.
In the implemented situation, we assume the size of ORIS is larger than the incident beam so that the ORIS operates in saturated power scaling regime~\cite{RefWorks:RefID:40-2024optical}.

\begin{figure}
\centering
\includegraphics[width=\linewidth]{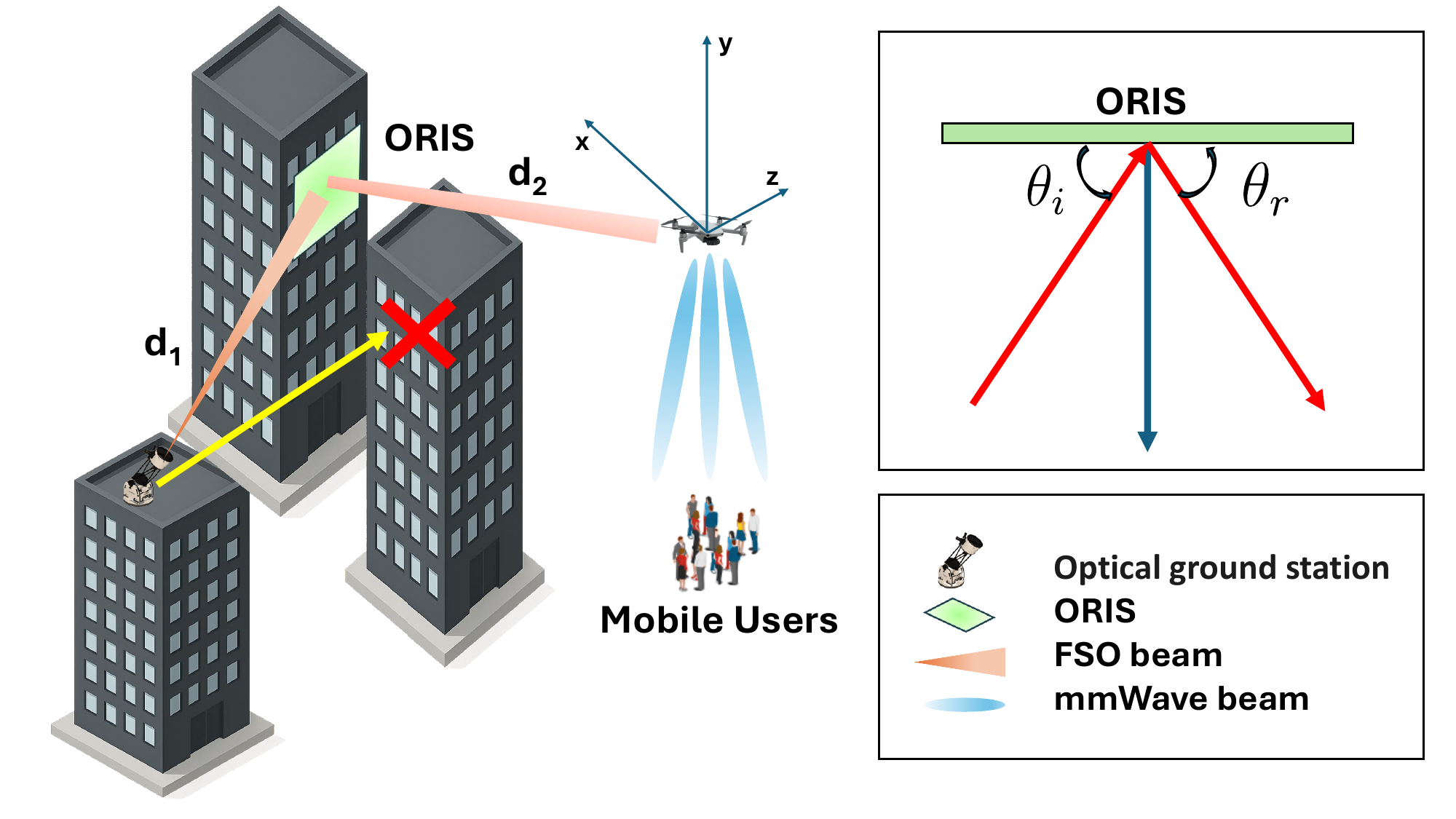}
\caption{Illustration of the SLIPT-enabled ground-to-UAV system in 6G mobile networks with UAV serving as an aerial base station. On the right is the view of ORIS from above.}
\label{fig:system_model}
\end{figure}

The Tx, ORIS, and Rx are located such that the distance $d_{1}$ between the Tx and the ORIS and the distance $d_{2}$ between the ORIS and the Rx have the relationship of $d_{1} = d_{2} = \frac{Z}{2}$.
Here, $Z$ is the total transmission distance.
The angle of incidence and reflection of the transmitted beam central ray at the ORIS are $\theta_{i}$ and $\theta_{r}$, respectively.
The system employs subcarrier intensity modulation with direct detection (SIM/DD). Information bits are first modulated using MPSK, and the resulting RF subcarrier is used to modulate the intensity of an optical subcarrier.
The modulated optical signal is expressed as
$P_t = S(B + x)$~\cite{RefWorks:RefID:9-2016optimal},
where $S$ is the slope efficiency of the laser diode, $x$ is the MPSK symbol with peak amplitude $A \in [0, \frac{I_H - I_L}{2}]$, and $B$ is a DC bias added to ensure non-negativity~\cite{RefWorks:RefID:3-2018simultaneous}. Here, $I_L$ and $I_H$ denote the minimum and maximum input bias currents, respectively.

The received optical power is given by $P_r = I P_t$, where $I$ is the combined channel coefficient.
Then, a photodiode (PD), integrated into the UAV Rx, converts this optical signal to an electrical current~\cite{RefWorks:RefID:6-2024slipt}
\begin{IEEEeqnarray}{rCl}
i_r &=& r I S B + r I S x + n \nonumber \\
&=& I'_{\text{DC}} + I_{\text{AC}} + n,
\end{IEEEeqnarray}
where $r$ is the PD responsivity, $n$ is zero-mean Gaussian noise, $I'_{\text{DC}} = r I S B$ is the DC component, and $I_{\text{AC}} = r I S x$ is the AC component~\cite{RefWorks:RefID:6-2024slipt}.

\subsection{SLIPT Strategies}
The Rx processes the incoming signal for both ID and EH using the TS-PS (TSPS) method, which is considered a practical method for SLIPT systems~\cite{RefWorks:RefID:6-2024slipt}.
In this method, TS and PS are combined wherein during the ID mode of the TS method, the DC component present in the transmitted signal is efficiently utilized for energy harvesting using a power splitter~\cite{RefWorks:RefID:6-2024slipt}.
The total DC component of the received electrical current used for EH is given by~\cite{RefWorks:RefID:9-2016optimal}
\begin{IEEEeqnarray}{rCl}
I_{\text{EH}} =
\begin{cases}
I'_{\text{DC}}, & \text{if DC component is used} \\
I'_{\text{DC}} + \zeta \bar{I}_{\text{AC}}, & \text{if AC, DC components are used}
\end{cases}
\end{IEEEeqnarray}
where $\zeta$ is the AC-to-DC conversion efficiency, $I'_{\text{DC}} = r I S B$ and $\bar{I}_{\text{AC}} = \mathbb{E}[I_{\text{AC}}] = r I S \mathbb{E}[x] = r I S A$.
\\\underline{ID Phase:}
During the ID phase, $T_{\text{ID}} = \tau T_{\text{tot}}$, where $\tau$ is the ratio of the ID phase duration to the total frame duration $T_{\text{tot}}$, the Tx sends both AC and DC components.
The received power is split into two streams by the factor $\rho$ simultaneously for ID and EH, written as
\begin{IEEEeqnarray}{rCL}
I_{\text{ID}} &=& \rho (\bar{I}_{\text{AC}} + n) \nonumber\\
I_{\text{EH}} &=& (1 - \rho)(I'_{\text{DC}} + \zeta \bar{I}_{\text{AC}}).
\label{eq:i_id_tsps}
\end{IEEEeqnarray}
where $I_{\text{ID}}$ and $I_{\text{EH}}$ are the currents for ID and EH, respectively and $\rho$ is the power splitting ratio.
\underline{EH Phase:}
During the EH phase with a duration of $T_{EH}$, only DC is transmitted (i.e., $A = 0$, $B = I_H$). Thus, $I_{\text{EH}} = I'_{\text{DC}}$.
Therefore, the total EH current over $T_{\text{tot}}$ is given by
$I_{\text{EH}} = (1 - \rho)(I'_{\text{DC}} + \zeta \bar{I}_{\text{AC}}) + I'_{\text{DC}}$~\cite{RefWorks:RefID:6-2024slipt}.

\section{Channel Model}
Upon reception, the optical signal is degraded by several independent channel impairments:
i) atmospheric path loss,
ii) turbulence-induced fading,
iii) pointing errors from UAV hovering,
iv) link interruption due to AoA fluctuations.
The combined channel model of the ground-to-UAV FSO system is given by
$I = h I_{\text{al}} I_{\text{pl}} I_{\text{AoA}}$,
where $h$ is the atmospheric path loss, $I_{\text{al}}$ is the turbulence-induced fading, $I_{\text{pl}}$ is the pointing error loss, and $I_{\text{AoA}}$ is the link interruption due to AoA fluctuations.

\subsection{Atmospheric Path Loss and Turbulence-Induced Fading}
For an optical link with a distance $d$ and an attenuation parameter $\zeta_1$, the atmospheric path loss is given by
$h = \exp(-\zeta_1 d)$~\cite{RefWorks:RefID:6-2024slipt}.
The optical turbulence between the Tx, ORIS, and the Rx can be modeled using the generalized Malaga distribution with the PDF written as
$f_{\text{al}}(I_{\text{al}}) = A_M\!\sum_{m=1}^{\beta}\!a_m I_{\text{al}}^{\frac{\alpha + m}{2} - 1} K_{\alpha - m}\left( 2 \sqrt{ \frac{\alpha \beta I_{\text{al}} }{g\beta + \Omega'}} \right)$~\cite{RefWorks:RefID:10-2016performance},
where
$A_M = \frac{2 \alpha^{\alpha/2}}{g^{1 + \alpha/2} \Gamma(\alpha)} \left( \frac{g \beta}{g \beta + \Omega'} \right)^{\beta + \alpha/2}$, 
$a_m = \binom{\beta - 1}{m - 1} \frac{(g \beta + \Omega')^{1 - \frac{m}{2}}}{(m - 1)!} \left( \frac{\Omega'}{g} \right)^{m - 1} \left( \frac{\alpha}{\beta} \right)^\frac{m}{2}$, and 
$g = 2 b_0 (1 - \rho_1)$. 
Furthermore, $\rho_1$ is the amount of scattering power, and $2 b_0$ is the average power of the total scatter component.
$K_v$ is the modified Bessel function of the second kind of order $v$.
$\alpha$ and $\beta$ are large-scale and small-scale scattering parameters~\cite{RefWorks:RefID:13-boluda-ruiz2016approximation},
\begin{IEEEeqnarray}{rCl}
\alpha &= \left[ \exp\left( {0.49 \sigma_R^2}/{\left(1 + 1.11 \sigma_R^{12/5}{}\right)^{7/6}} \right) - 1 \right]^{-1} \nonumber\\
\beta  &= \left[ \exp\left( {0.51 \sigma_R^2}/{\left(1 + 0.69 \sigma_R^{12/5}{}\right)^{5/6}} \right) - 1 \right]^{-1}, \nonumber
\end{IEEEeqnarray}
where $\sigma_R^{2}{} = 1.23 C_n^{2}{} k^{7/6}{}d^{11/6}$ is the Rytov variance of the refractive index structure parameter.
Here, $C_n^{2}{}$ is the refractive index structure parameter, $k=2 \pi / \lambda$ is the wave number, and $d$ is the propagation distance in meters~\cite{RefWorks:RefID:13-boluda-ruiz2016approximation}.
Since we assume the equal distance for the Tx-ORIS and ORIS-Rx paths, similar propagation altitudes, and a passive ORIS reflection, the overall turbulence-induced fading can be approximated by a single equivalent channel coefficient.

\subsection{Pointing Errors with ORIS}
Pointing errors occur due to the radial displacement vector from the center of the beam and receiving aperture $r_d = [x_d, y_d]$, which results from two error vectors~\cite{RefWorks:RefID:16-2022comprehensive}:
i) Displacement vector due to Tx's position deviation due to the building sway $\mathbf{r}_t = [x_t, y_t]$, and
ii) Displacement vector due to Rx's position deviation due to UAV hovering $\mathbf{r}_r = [x_r, y_r]$.
Hence, the total displacements along the $x$ (horizontal) and $y$ (vertical) axes are, respectively, given by $x_d = x_t + x_r$ and $y_d = y_t + y_r$.
Since the position and orientation deviations occur due to random events, they are assumed to follow Gaussian distributions as per the central limit theorem.
In the most general case, $x_t$, $x_r$, $y_t$, and $y_r$ are non-zero-mean Gaussian random variables (RVs), i.e.,
$x_t \sim \mathcal{N}(\mu_{txp}, \sigma_{txp}^2)$, $x_r \sim \mathcal{N}(\mu_{rxp}, \sigma_{rxp}^2)$, $y_t \sim \mathcal{N}(\mu_{typ}, \sigma_{typ}^2)$, $y_r \sim \mathcal{N}(\mu_{ryp}, \sigma_{ryp}^2)$.
Consequently, the displacement vector also becomes a non-zero-mean Gaussian RV, i.e., $x_d \sim \mathcal{N}(\mu_x, \sigma_x^2), y_d \sim \mathcal{N}(\mu_y, \sigma_y^2)$,
where $\mu_x = \mu_{txp} + \mu_{rxp}$ and $\mu_y = \mu_{typ} + \mu_{ryp}$ are the mean values and $\sigma_x^2 = \sigma_{txp}^2 + \sigma_{rxp}^2$ and $\sigma_y^2 = \sigma_{typ}^2 + \sigma_{ryp}^2$ are the variances of the displacement vector in $x$ and $y$ directions, respectively.
Under the assumption that the random displacements are independent and non-identical, the total radial displacement is characterized by the Beckmann distribution but can be conveniently approximated by a modified Rayleigh distribution with the PDF given by~\cite{RefWorks:RefID:13-boluda-ruiz2016approximation}
\begin{IEEEeqnarray}{rCl}
f_{r_d}(r_d) = \frac{r_d}{2 \sigma_m^2} \exp\left(-\frac{r_d^2}{2 \sigma_m^2}\right), r_d \ge 0,
\label{eq:perayleighpdf}
\end{IEEEeqnarray}
where the total displacement variance $\sigma_m$ can be expressed as
\begin{IEEEeqnarray}{rCl}
\sigma_m^2 &=& \left(\frac{3\mu_x^2\sigma_x^4 + 3\mu_y^2\sigma_y^4 + \sigma_x^6 + \sigma_y^6}{2}\right)^\frac{1}{3}.
\label{eq:displacementvariance}
\end{IEEEeqnarray}
Note that \eqref{eq:perayleighpdf} and \eqref{eq:displacementvariance} represent a more general pointing-error model than (18) and (19) in~\cite{RefWorks:RefID:6-2024slipt}.
(19) in~\cite{RefWorks:RefID:6-2024slipt} is also inaccurate, following the approach in~\cite{RefWorks:RefID:13-boluda-ruiz2016approximation}.

The instantaneous loss due to radial displacement, considering the optical beam after being reflected by the ORIS with QPS profile, can be expressed as~\cite{RefWorks:RefID:43-2025optical}
\begin{IEEEeqnarray}{rCl}
I_{\text{pl}} = A_x A_y \exp\left(-\frac{2 x_d^2}{w_{rx,x(eq)}^{QPS}{}^2} - \frac{2 y_d^2}{w_{rx,y(eq)}^{QPS}{}^2}\right),
\label{eq:pointingerrorloss}
\end{IEEEeqnarray}
where the equivalent beam widths at the Rx for the ORIS are given by~\cite{RefWorks:RefID:43-2025optical} as $w_{rx,x(eq)}^{QPS}{} = w_{rx,x_d}^{QPS}{} \sqrt{\frac{\sqrt{\pi}A_x}{2 \nu_x \exp(-\nu_x^2)}}$ and $w_{rx,y(eq)}^{QPS}{} = w_{rx,y_d}^{QPS}{} \sqrt{\frac{\sqrt{\pi}A_y}{2 \nu_y \exp(-\nu_y^2)}}$.
$A_x = erf \left(\frac{r_a \sqrt{\pi}}{\sqrt{2} w_{rx,x}^{QPS}}\right)$ and $A_y = erf \left(\frac{r_a \sqrt{\pi}}{\sqrt{2} w_{rx,y}^{QPS}}\right)$ are the area of the receiving aperture in $x$ and $y$ directions, respectively, and $r_a$ is the radius of the receiving aperture.
The beam widths at the Rx is now expressed as~\cite{RefWorks:RefID:43-2025optical},~\cite{RefWorks:RefID:13-boluda-ruiz2016approximation}
\begin{IEEEeqnarray}{rCl}
w_{rx,x}^{QPS}{} &=& w(d_1) \frac{\left|\sin \theta_r\right|}{\left|\sin \theta_i\right|} \sqrt{\epsilon \left(\frac{\sin ^2 \theta_i}{\sin^2 \theta_r} \Lambda_1\right)^2 + \left(\frac{d_2}{2 f}\right)^2} \nonumber \\
w_{rx,y}^{QPS}{} &=& w(d_1) \sqrt{\epsilon \Lambda_1^2 + \left(\frac{d_2}{2 f}\right)^2},
\end{IEEEeqnarray}
where for a Gaussian beam propagating in atmospheric turbulence, the beam footprint at the ORIS can be expressed as~\cite{RefWorks:RefID:14-2007outage}
$w(d_1) = w_0 \left[1 + \epsilon \left(\frac{\lambda d_1}{\pi w_0^{2}{}}\right)\right] ^\frac{1}{2}$
where $w_0$ is the beam waist radius at the Tx, $\epsilon= (1 + 2 w_0^2/\rho_0^{2}{}(z))$ is the pointing error, $\rho_{0}(z) = (0.55 C_n^{2}{} k^2 z) ^{-3/5}$ is the coherent length.
Here, $f$ is the focal length of the ORIS. When $f=\frac{d_2}{2}$, it is equivalent to the LPS profile, which redirects the beam to the Rx without focusing it~\cite{RefWorks:RefID:43-2025optical},~\cite{RefWorks:RefID:40-2024optical}.

From \eqref{eq:perayleighpdf} and \eqref{eq:pointingerrorloss}, the PDF of the pointing error can be expressed as
$f_{I_{\text{pl}}}(I_{\text{pl}}) = {\zeta_{\text{mod}}^2} I_{\text{pl}}^{\zeta_{mod}^2 - 1}/{(A_{mod})^{\zeta_{mod}^2}}$~\cite{RefWorks:RefID:13-boluda-ruiz2016approximation},
where $\zeta_{\text{mod}} = {\sqrt{w_{rx,x(eq)}^{QPS}{} w_{rx,y(eq)}^{QPS}{}}}/{2 \sigma_m}$ and $A_{mod} = A_x A_y \Phi$
and $\Phi = \exp\left(\frac{1}{\zeta_{\text{mod}}^2} - \frac{1}{2 \varphi_x^2} - \frac{1}{2 \varphi_y^2} - \frac{\mu_x^2}{2\sigma_x^2\varphi_x^2} - \frac{\mu_y^2}{2 \sigma_y^2\varphi_y^2}\right)$, $\varphi_x = {w_{rx,x(eq)}^{QPS}{}}/{2 \sigma_m}$, and $\varphi_y = {w_{rx,y(eq)}^{QPS}{}}/{2 \sigma_m}$.

\subsection{AoA Fluctuations}
AoA fluctuation is a phenomenon that occurs when the UAV's orientation deviates from the normal line to the PD because of the UAV's hovering instability.
This deviation can cause the received beam to be outside the field-of-view (FOV) of the detector, leading to the loss of a signal.
This phenomenon can be modeled using the link interruption coefficient $I_{\text{AoA}}$, written as
\begin{IEEEeqnarray}{rCl}
I_{\text{AoA}} =
\begin{cases}
1, & \text{if } \theta_a \leq \theta_{\text{FOV}} \\
0, & \text{if } \theta_a \geq \theta_{\text{FOV}}
\end{cases}
\label{eq:AoAindicator}
\end{IEEEeqnarray}
where $\theta_{\text{FOV}}$ is the UAV's FOV and $\theta_a = \sqrt{\theta_{rx}^2 + \theta_{ry}^2}$ is the AoA.
$\theta_{rx}$ and $\theta_{ry}$ are Gaussian RVs such that $\theta_{rx} \sim \mathcal{N}(\theta'_{rx}, \sigma_{rxa}^2)$ and $\theta_{ry} \sim \mathcal{N}(\theta'_{ry}, \sigma_{rya}^2)$ where $\theta'_{rx}$ and $\theta'_{ry}$ are the UAVs boresight angles in x and y directions, respectively, and $\sigma_{rxa}$ and $\sigma_{rya}$ are the SD of the UAV orientation in the $x-z$ and $y-z$ planes, respectively.

The PDF of $\theta_a$ is modeled by Beckmann distribution, which is approximated using a modified Rayleigh distribution as
\begin{IEEEeqnarray}{rCl}
f_{\theta_a}(\theta_a) \approx \frac{\theta_a}{\sigma_a^2} \exp\left(-\frac{\theta_a^2}{2 \sigma_a^2}\right), \quad \theta_a \geq 0,
\label{eq:thetaa_pdf}
\end{IEEEeqnarray}
where
\begin{IEEEeqnarray}{rCl}
\sigma_a^2 = \left( \frac{3 \theta_{rx}'^2 \sigma_{rxa}^4 + 3 \theta_{ry}'^2 \sigma_{rya}^4 + \sigma_{rxa}^6 + \sigma_{rya}^6}{2} \right)^{1/3}.
\end{IEEEeqnarray}
Using \eqref{eq:AoAindicator} and \eqref{eq:thetaa_pdf}, the PDF of AoA fluctuations $f_{I_{\text{AoA}}}(I_{\text{AoA}})$ is given as
$f_{I_{\text{AoA}}}(I_{\text{AoA}}) = a_1 \delta(I_{\text{AoA}}) + (1 - a_1) \delta(I_{\text{AoA}} - 1)$,
where $a_1 = \exp\left(-\frac{\theta_{\text{FOV}}^2}{2 \sigma_a^2}\right)$ and $\delta(\cdot)$ is the Dirac delta function.

\subsection{Combined Channel Model}
The PDF of the combined channel model of the ground-to-UAV FSO system is given by~\cite{RefWorks:RefID:16-2022comprehensive}
\begin{IEEEeqnarray}{rCl}
f_I(I) &\approx& a_1 \delta(I) + \frac{(1 - a_1) \zeta_{\text{mod}}^2 A_M}{2I} \sum_{m=1}^{\beta} b_m \nonumber\\
&&\times G_{1,3}^{3,0} \left( \left. \frac{B_1 I}{h A_{\text{mod}}} \right| \begin{array}{c}
1 + \zeta_{\text{mod}}^2 \\
\zeta_{\text{mod}}^2, \alpha, m
\end{array} \right),
\label{eq:combinedpdf}
\end{IEEEeqnarray}
where $B_1 = \frac{\alpha \beta}{g \beta + \Omega'}$, $\Omega' = \Omega + 2 b_0 \rho_1 + 2 2\sqrt{2 b_0 \rho_1 \Omega} \cos(\phi_a - \phi_b)$ where $\Omega$ is the average power of LoS component.
$b_m = a_m B^{-\frac{\alpha + m}{2}}$, and $G^{m,n}_{p,q}(z)$ is the Meijer-G function.
The instantaneous SNR is given as $\gamma = \frac{\bar{\gamma} I^2}{(k_1 A_{\text{mod}} h)^2}$ where $k_1 = \frac{\zeta_{\text{mod}}^2}{\zeta_{\text{mod}}^2 + 1}(g + \Omega')(1 - a_1)$.

\begin{lemma}
\label{Lemma1}
The PDF of the instantaneous SNR is obtained as
\begin{IEEEeqnarray}{rCl}
f_\gamma(\gamma) &\approx& a_1 \delta\left( \sqrt{\frac{\gamma}{\bar{\gamma}}} k_1 A_{\text{mod}} h \right)
+ (1 - a_1) \frac{\zeta_{\text{mod}}^2 A_M}{4 \gamma} \nonumber\\
&&\times \sum_{m=1}^{\beta} b_m G_{1,3}^{3,0} \left( \left. k_1 B_1 \sqrt{\frac{\gamma}{\bar{\gamma}}} \right| \begin{array}{c}
\zeta_{\text{mod}}^2 + 1 \\
\zeta_{\text{mod}}^2, \alpha, m
\end{array} \right).
\label{eq:snrpdf}
\end{IEEEeqnarray}
\end{lemma}
\begin{proof}
See Appendix \ref{appendix_C}.
\end{proof}

\begin{lemma}
\label{Lemma2}
The CDF of the instantaneous SNR is given by
\begin{IEEEeqnarray}{rCl}
F_\gamma(\gamma) &\approx& a_1 + (1 - a_1) \frac{\zeta_{\text{mod}}^2 A_M}{8\pi} \nonumber\\
&&\!\times\!\sum_{m=1}^{\beta}\!b_m 2^{\alpha + m - 1}\!
G_{3,7}^{6,1}\!\left(\!\left. \frac{(k_1 B_1)^2 \gamma}{16 \bar{\gamma}} \right|
\begin{array}{c}\!
1, K_2 \\
K_3, 0
\end{array}\!\right)\!,
\label{eq:snrcdf}
\end{IEEEeqnarray}
where $K_2\!=\!\left[\!\frac{\zeta_{\text{mod}}^2 + 1}{2}, \frac{\zeta_{\text{mod}}^2 + 2}{2}\!\right]$, $K_3 = \left[ \frac{\zeta_{\text{mod}}^2}{2}, \frac{\zeta_{\text{mod}}^2 + 1}{2}, \frac{\alpha}{2}, \frac{m}{2}, \frac{m + 1}{2} \right]$.
\end{lemma}
\begin{proof}
See Appendix \ref{appendix_D}.
\end{proof}

\begin{remark}
The parameters $K_2$ and $K_3$ in Lemma~\ref{Lemma2} incorporate the effect of ORIS, and are generalized to pointing errors with a non-zero mean compared to~\cite{RefWorks:RefID:6-2024slipt}.
\end{remark}

\section{Performance Analysis}
\label{performance}
\subsection{Harvested Energy}
\begin{corollary}
\label{corollary1}
In the TSPS method, the total energy harvested is given by
\begin{IEEEeqnarray}{rCl}
E_{TSPS} &=& \tau E_1 \sum_{m=1}^{\beta} b_m \text{G}_{3,5}^{5,1}\left( C \middle| \begin{array}{c} K_4 \\ K_5 \end{array} \right) \nonumber\\
&&\hspace{-1cm} +(1 - \tau)E_1\sum_{m=1}^{\beta}b_m\text{G}_{3,5}^{5,1}\left(\frac{B_1 I_0}{r h S I_H A_{\text{mod}}}\middle|\begin{array}{c}K_4\\K_5\end{array}\right),
\end{IEEEeqnarray}
where
$C = \frac{B_1 I_0}{h A_{\text{mod}} (1 - \rho)(B + \zeta A) r S}$, $E_1 = T_{tot} F V_t I_0 (1 - a_1) \frac{\zeta_{\text{mod}}^2 A_M}{2}$, $K_4 = \left[-1, 0, 1 + \zeta_{\text{mod}}^2\right]$, $K_5 = \left[\zeta_{\text{mod}}^2, \alpha, m, -1, -1\right]$, and
$b_m = a_m B^{-\frac{\alpha+m}{2}}$, $B_1 = \frac{\alpha \beta}{g\beta + \Omega'}$.
\end{corollary}
\begin{proof}
See Appendix \ref{appendix_A}.
\end{proof}
\begin{remark}
The parameters $K_4$ and $K_5$ incorporate the effect of ORIS and generalized pointing errors with a non-zero mean.
Here, the incorrect parameter order is corrected for $K_4$ and $K_5$ as compared to (33) in~\cite{RefWorks:RefID:6-2024slipt}.
\end{remark}

\subsection{Outage Probability}
\begin{corollary}
\label{Corollary2}
The outage probability is expressed as
\begin{equation}
P_{out}(\gamma_{th}) = F_\gamma(\gamma_{th}),
\end{equation}
where $F_\gamma(\gamma)$ is the CDF of the SNR given in Lemma~\ref{Lemma2}.
\end{corollary}
\begin{proof}
See Appendix \ref{appendix_OP}.
\end{proof}

\subsection{SER Calculation}
\begin{corollary}
\label{Corollary3}
The SER for the ground-to-UAV ORIS-assisted FSO system using MPSK modulation can be derived as
\begin{IEEEeqnarray}{rCl}
SER &=& \tau\left(\frac{Aa_1}{2} + \frac{AM\zeta_{\text{mod}}^2 A(1 - a_1)}{16 \pi} \times \sum_{m=1}^{\beta} b_m 2^{\alpha+m-1} \right. \nonumber\\
&& \left. \times G_{4,7}^{6,2} \Bigg( \frac{(k_1 B_1)^2}{\sin^2(\pi/M) 16 \bar{\gamma}} \;\middle|
\begin{array}{c}
1, 0.5, K_2 \\
K_3, 0
\end{array} \Bigg) \right).
\end{IEEEeqnarray}
\end{corollary}
\begin{proof}
See Appendix \ref{appendix_SER}.
\end{proof}

\section{Numerical Analysis}
For a fair comparison with the case without ORIS, the total distance is equivalent in both cases, i.e., $d = d_1 + d_2 = Z$.
The channel and system parameters are summarized in Table~\ref{tab:channel_parameters}.

\begin{table}
{\footnotesize\caption{Channel and System Parameters}
\label{tab:channel_parameters}
}
\centering
\scalebox{0.68}{
\begin{tabular}{|c|c|c||c|c|c|}
\hline
\textbf{Parameter} & \textbf{Value} & \textbf{Ref.} & \textbf{Parameter} & \textbf{Value} & \textbf{Ref.} \\
\hline
$T_{tot}$ & 1 & ~\cite{RefWorks:RefID:3-2018simultaneous} & $\rho_1$ & 0.596 & ~\cite{RefWorks:RefID:10-2016performance}\\
\hline
$F$ & 0.75 & ~\cite{RefWorks:RefID:17-1994solar} & $\Omega$ & 1.3265 & ~\cite{RefWorks:RefID:10-2016performance} \\
\hline
$I_0$ & $10^{-9}$ A & ~\cite{RefWorks:RefID:17-1994solar} & $b_0$ & 0.1079 & ~\cite{RefWorks:RefID:10-2016performance} \\
\hline
$V_t$ & 25 mV & ~\cite{RefWorks:RefID:17-1994solar} & $\phi_a - \phi_b$ & $\pi/2$ & ~\cite{RefWorks:RefID:10-2016performance} \\
\hline
$r$ & 0.6 A/W & ~\cite{RefWorks:RefID:17-1994solar} & $\sigma$ & $10^{-14}$ & ~\cite{RefWorks:RefID:18-2009bitinterleaved}  \\
\hline
$S$ & 1.33 W/A & ~\cite{RefWorks:RefID:9-2016optimal} & $\sigma_a$ & 3.45 mrad & ~\cite{RefWorks:RefID:16-2022comprehensive}\\
\hline
$I_L$ & 200 mA & ~\cite{RefWorks:RefID:9-2016optimal} & $d$ & 200 m & ~\cite{RefWorks:RefID:6-2024slipt} \\
\hline
$I_H$ & 1200 mA & ~\cite{RefWorks:RefID:9-2016optimal} & M & 4 & -  \\
\hline
$\lambda$ & 1550 nm & ~\cite{RefWorks:RefID:16-2022comprehensive} & $\rho$ & 0.5 & -\\
\hline
$\zeta$ & 0.3 & ~\cite{RefWorks:RefID:16-2022comprehensive} & $\tau$ & 0.5 & - \\
\hline
$r_a$ & 5 cm & ~\cite{RefWorks:RefID:16-2022comprehensive} & $d_1$ & 100 m & - \\
\hline
$\sigma_{tx}, \sigma_{rx}$ & 40 cm & ~\cite{RefWorks:RefID:16-2022comprehensive}& $d_2$ & 100 m & -  \\ 
\hline
$\sigma_{ty}, \sigma_{ry}$ & 30 cm & ~\cite{RefWorks:RefID:16-2022comprehensive}& $\theta_i, \theta_r $ & $\pi/4$ rad & -  \\ 
\hline
$\zeta_1$ & 10.2 dB/km (Heavy Rain) & \cite{RefWorks:RefID:18-2009bitinterleaved} & $\mu_x$ & 0.03 m & - \\
\hline
$C_n^2$ & $1.7 \times 10^{-13}$ & ~\cite{RefWorks:RefID:19-2017optical} & $\mu_y$ & 0.05 m & - \\ 
\hline
\end{tabular}}
\end{table}

\begin{figure}
\centering
\includegraphics[width=\linewidth]{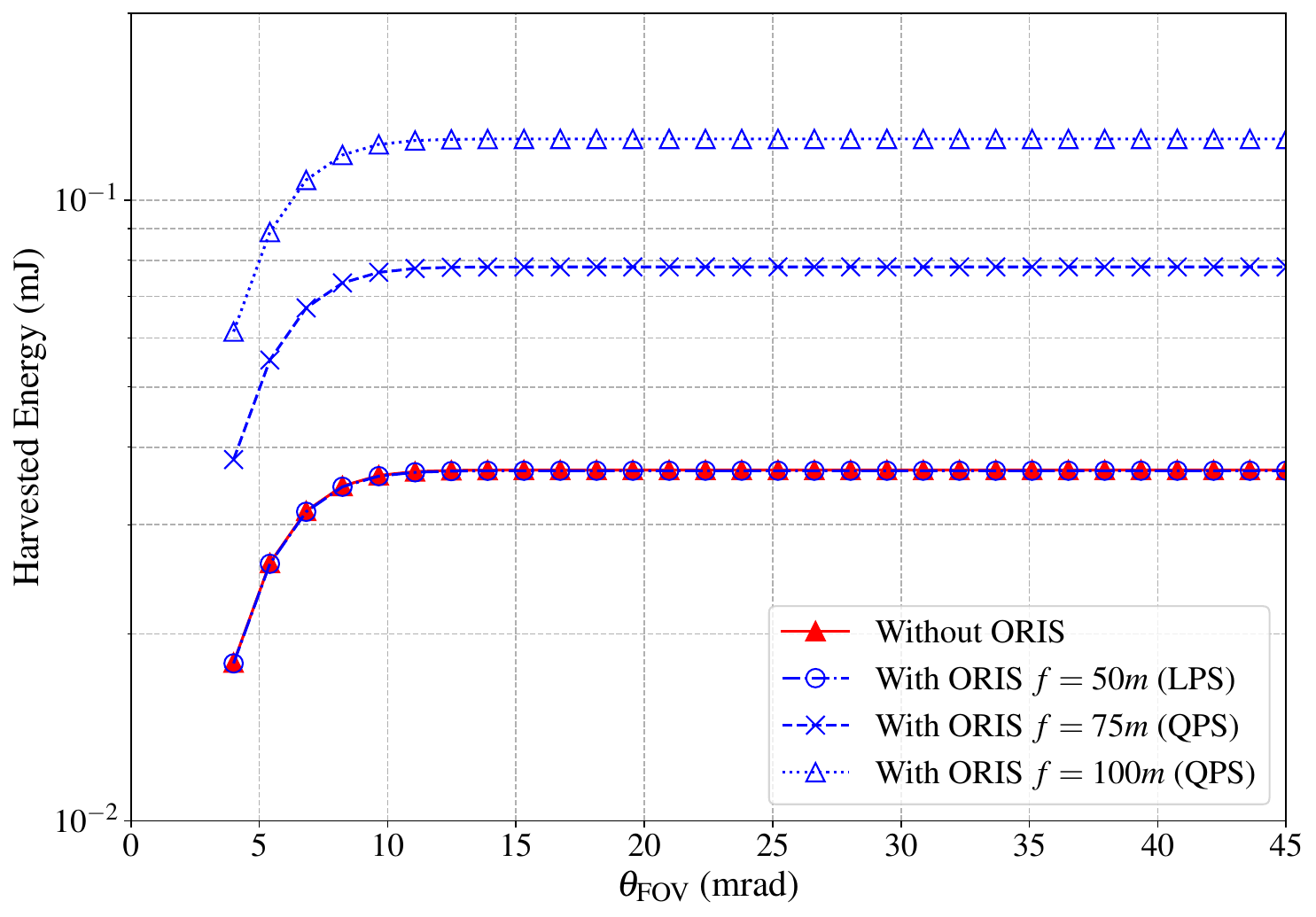}
\caption{Harvested energy versus $\theta_{FOV}$ using ORIS with different focus distances.}
\label{eh_w_ris}
\end{figure}

In Fig.~\ref{eh_w_ris}, the harvested energy is plotted against the Rx FOV when $\sigma_a = 3.45$ mrad.
The corresponding beam widths at the Rx aperture when the focus distance $f=50, 75, 100$ m are $w_{rx,x}^{QPS} = w_{rx,y}^{QPS} = 2.99, 1.99, 1.49$ m, respectively.
The harvested energy increases with the increment in the focus distance $f$ of the QPS profile implemented on the ORIS.
This is because a larger focus distance allows the beam to be more focused at the Rx, which results in a higher received power.
However, it is crucial that the focus distance should not be too large, as it can lead to a decrease in the outage probability, which is further highlighted in Fig.~\ref{outage_vs_snr}.
As seen in Fig.~\ref{outage_vs_snr}, the outage probability increases for larger focus distances.
This is because the narrower beam width increases the system's sensitivity to pointing errors, resulting in a higher outage probability.
\begin{figure}
\centering
\includegraphics[width=\linewidth]{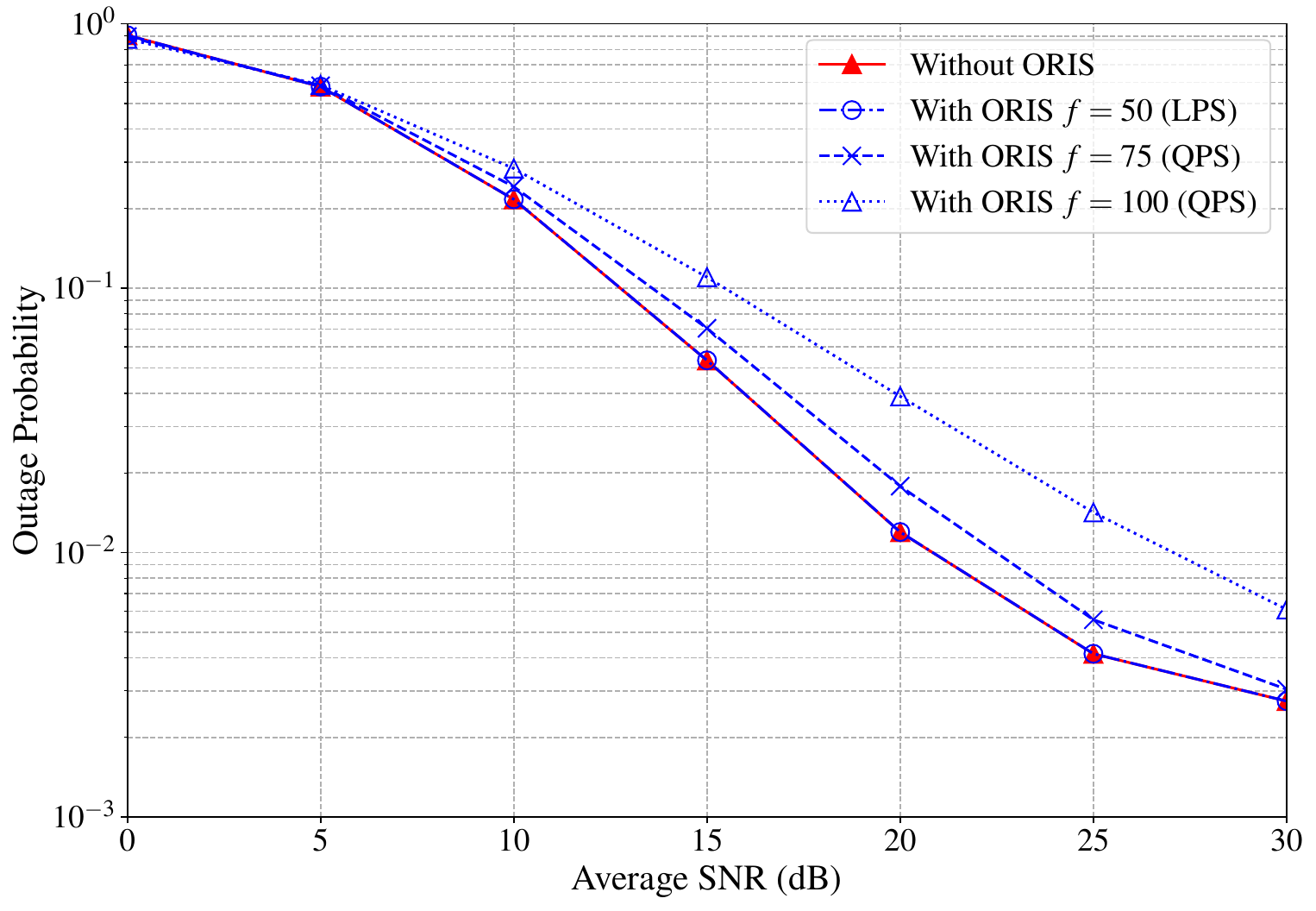}
\caption{Outage probability versus average SNR for different focus distances of the ORIS when $\gamma_{th} = 5$ dB and $\theta_{FOV} = 12$ mrad.}
\label{outage_vs_snr}
\end{figure}

 To better understand the energy-outage trade-off, Fig.~\ref{eh_vs_theta_f} shows the corresponding harvested energy for the same range of FOV angles and focus distances.
 EH is plotted with focus distances that yield acceptable outage performance from general communication reliability perspectives to business applications defined in International Telecommunications Union – Telecommunications Sector (i.e., $10^{-2}$ to $10^{-5}$ outage probability)~\cite{ITU-GSup51},~\cite{ITU-L163}.
 We note that although higher focus distances improve the energy harvesting due to tighter beam focusing, this improvement comes at the cost of increased outage probability, especially when the Rx's FOV is narrow.
 Conversely, lower focus distances maintain more robust communications but offer lower energy harvesting.
 These results suggest that the design of RIS-assisted SLIPT systems must consider a balance between energy efficiency and communication reliability.

 \begin{figure}
 \centering
 \includegraphics[width=.9\linewidth]{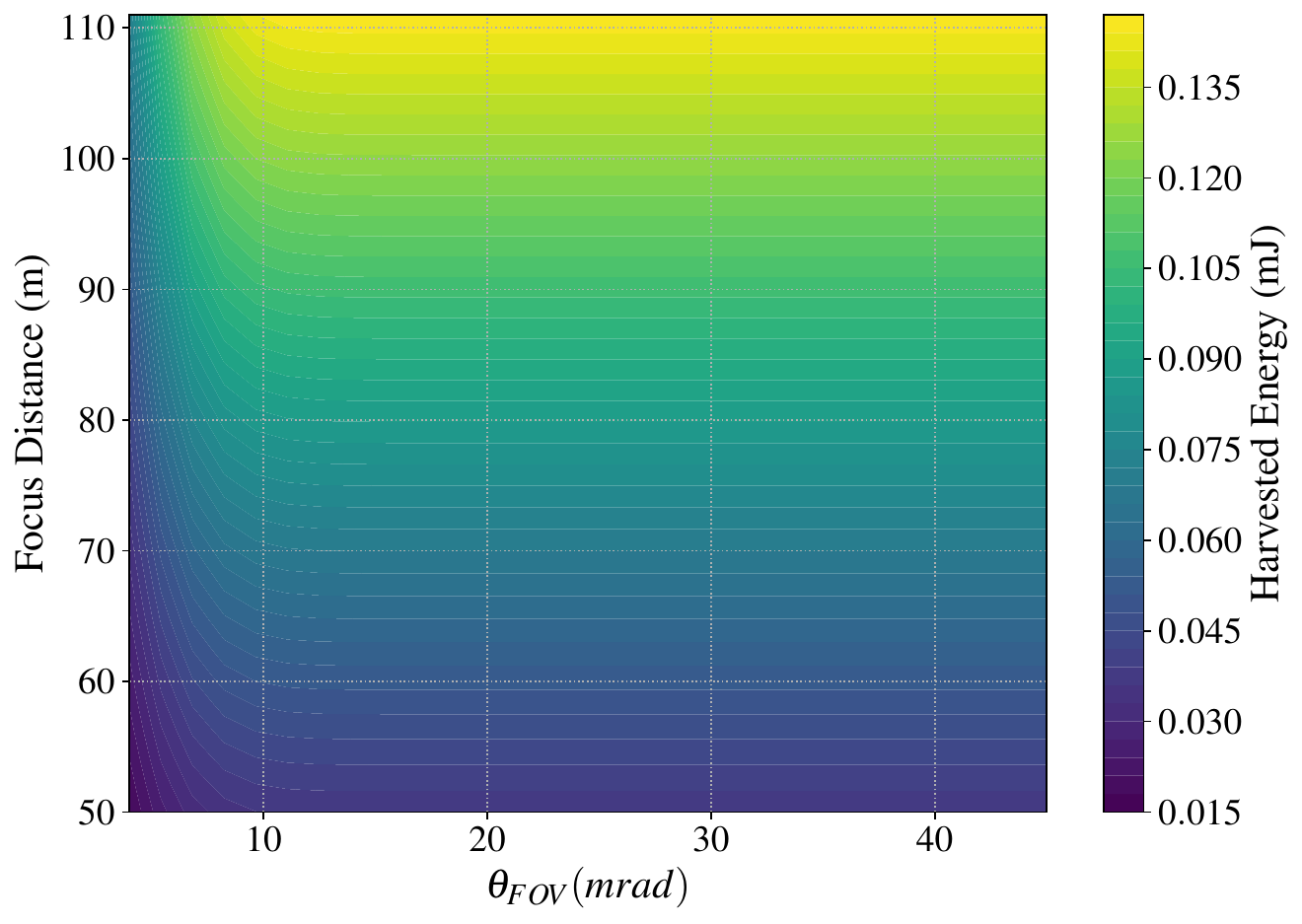}
 \caption{Harvested energy versus $\theta_{FOV}$ for different focus distances.}
 \label{eh_vs_theta_f}
 \end{figure}

Fig.~\ref{fig:ser_vs_snr} illustrates the SER as a function of average SNR for different configurations of focus distance $f$ and $\theta_{\text{FOV}}$.
As expected, the SER decreases with increasing average SNR due to improved signal quality.
However, the rate of decay and the SER floor are significantly influenced by the $f$ and $\theta_{\text{FOV}}$ settings.
Notably, for a fixed FOV, increasing the focus distance consistently degrades the SER performance.
This is due to tighter beam concentration at higher focus distances, which increases susceptibility to pointing errors and AoA misalignments.
While higher focus distances improve harvested energy, they increase the SER floor under tighter FOV conditions.
The results in Fig.~\ref{outage_vs_snr} and \ref{fig:ser_vs_snr} highlight the importance of jointly optimizing the ORIS configuration and UAV detector design.
For mission-critical communications requiring high reliability, conservative focus distances paired with wider FOVs are preferable.
Conversely, for energy-centric tasks with more relaxed reliability demands, aggressive focusing can be adopted to maximize SLIPT efficiency.

\begin{figure}
\centering
\includegraphics[width=\linewidth]{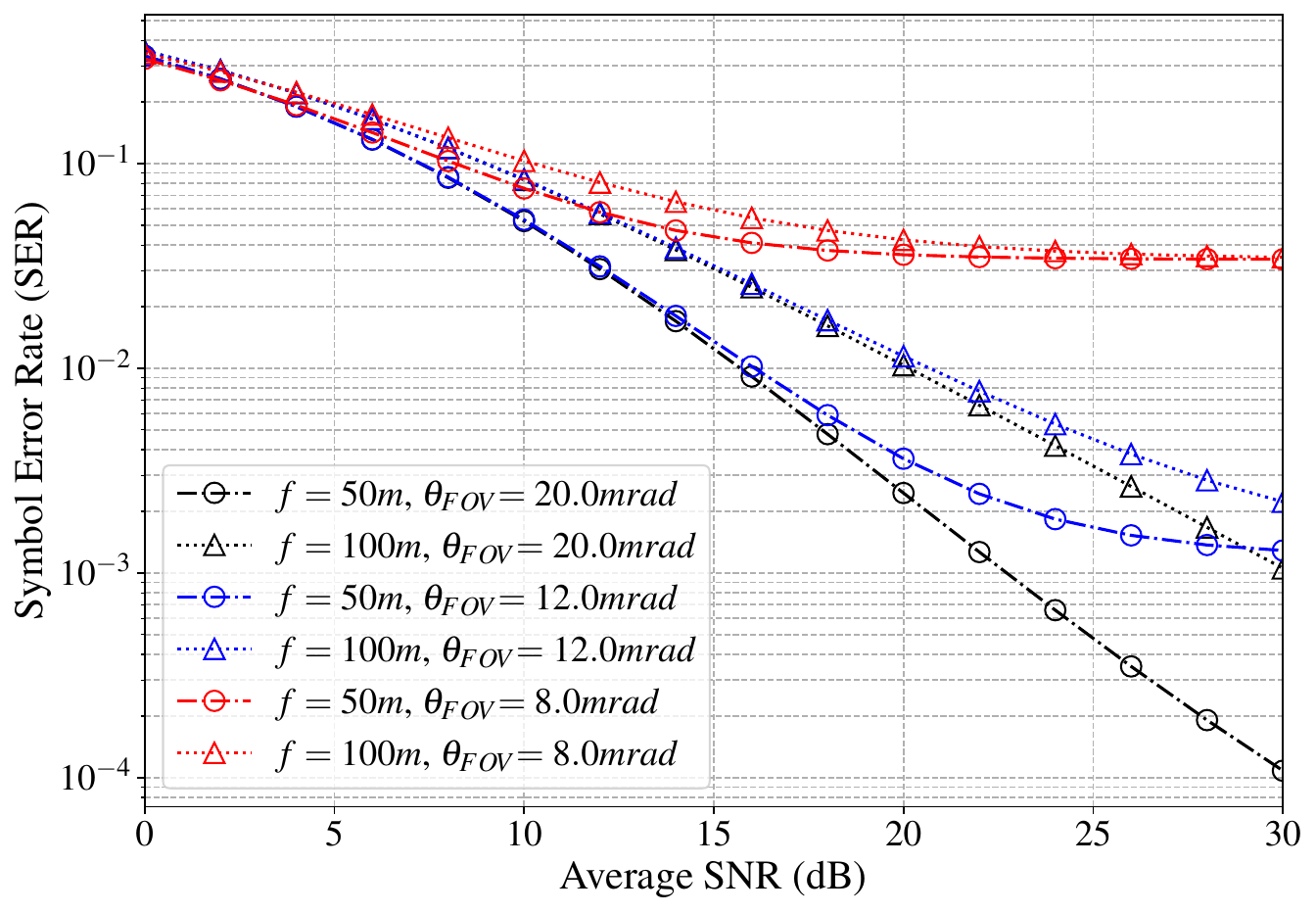}
\caption{SER versus average SNR for various focus distances and UAV FOV angles.}
\label{fig:ser_vs_snr}
\end{figure}

\section{Conclusions}
\label{conclusion}
In this paper, we introduced an ORIS-assisted ground-to-UAV FSO system that enables SLIPT to a UAV.
The use of an ORIS compensates for the lack of direct LoS by dynamically redirecting the laser beam, improving EH while maintaining communication reliability.
We derived a detailed channel model considering the use of ORIS that includes atmospheric attenuation, turbulence, pointing errors, and AoA fluctuations, and provided closed-form expressions for harvested energy, outage probability, and SER under the TSPS method.
Numerical simulations confirmed the benefits of ORIS, particularly in enhancing harvested energy while trading off with slight increases in outage probability and SER under high focus distances.

\appendices
\section{Proof of Lemma~\ref{Lemma1}}
\label{appendix_C}
The PDF of the instantaneous SNR in \eqref{eq:snrpdf} can be derived by substituting the expression for the instantaneous SNR into the combined channel PDF \eqref{eq:combinedpdf}.
This completes the proof.

\section{Proof of Lemma~\ref{Lemma2}}
\label{appendix_D}
The CDF of the SNR $F_\gamma(\gamma)$, can be calculated as $F_\gamma(\gamma) = \int_0^\gamma f_\gamma(x) dx$
where $f_\gamma(x)$ is the PDF of the SNR given in Lemma~\ref{Lemma1}.
This completes the proof.

\section{Proof of Corollary~\ref{corollary1}}
\label{appendix_A}
The maximum power of the solar panel, which we use for EH of the SLIPT at Rx, is given by~\cite{RefWorks:RefID:3-2018simultaneous}
\begin{equation}
P_{MPP} = FI_{DC}V_{OC},
\end{equation}
where $V_{OC} = V_t \ln(1 + \frac{I_{DC}}{I_0})$ is the open circuit voltage, $F$ is the fill factor, $V_t$ is the thermal voltage, $I_0$ is the reverse saturation current, and $I_{DC}$ is the current through the panel. Multiplying the maximum power by the time duration $T_{EH}$ and replacing $I_{DC}$ with $I_{EH}$, the harvested energy can be written as
\begin{equation}
E = T_{EH}FI_{EH}V_t \ln\left(1 + \frac{I_{EH}}{I_0}\right),
\label{eq:harvested_energy}
\end{equation}
where $T_{EH}$ is the energy harvesting time, $I_{EH}$ is the harvesting current.
The TSPS method is a combination of the TS and PS methods, and the energy is harvested in both the EH and ID modes of the TSPS.
Thus, $E_{TSPS}$ can be expressed using EH in TS mode ($E_{TS}$) and PS mode ($E_{PS}$) as
\begin{IEEEeqnarray}{rCL}
E_{TSPS} = \tau E_{PS} + E_{TS},
\label{eq:harvested_energy_tsps}
\end{IEEEeqnarray}
$E_{PS}$ and $E_{TS}$ are the derived below.
In the EH mode, $T_{EH} = (1-\tau)T_{tot}$. By substituting $I'_{\text{DC}}$ into \eqref{eq:harvested_energy} and averaging it over the PDF in \eqref{eq:combinedpdf}, we arrive at
\begin{equation}
E_{TS} = \int_0^\infty\!(1 - \tau)T_{tot}\!FrISBV_t \ln\left(\!1 + \frac{rIhSI_H}{I_0}\right)f_I(I) dI.
\end{equation}
Using $\ln(1 + x) = \text{G}_{2,2}^{1,2}\left( x \middle| \begin{array}{cc} 1, 1 \\ 1, 0 \end{array} \right)$, the closed-form expression becomes
\begin{equation}
E_{TS} = (1 - \tau) E_1 \sum_{m=1}^{\beta} b_m \text{G}_{3,5}^{5,1}\left(\frac{B_1 I_0}{r h S I_H A_{\text{mod}}} \middle| \begin{array}{c} K_4 \\ K_5 \end{array} \right).
\label{eq:harvested_energy_ts}
\end{equation}
For the ID mode, $T_{ID} = \tau T_{tot}$. By substituting $I'_{\text{DC}}$, $\bar{I}_{\text{AC}}$ and \eqref{eq:i_id_tsps} into \eqref{eq:harvested_energy} and averaging over the PDF in \eqref{eq:combinedpdf}
\begin{equation}
E_{PS} = E_1 \sum_{m=1}^{\beta} b_m \text{G}_{3,5}^{5,1}\left(\frac{B_1 I_0}{h A_{\text{mod}} (1 - \rho)(B + \zeta A) r S} \middle| \begin{array}{c} K_4 \\ K_5 \end{array} \right).
\label{eq:harvested_energy_ps}
\end{equation}
Thus, by substituting~\eqref{eq:harvested_energy_ps} and~\eqref{eq:harvested_energy_ts} into~\eqref{eq:harvested_energy_tsps}, we obtain Corollary~\ref{corollary1}.
This completes the proof.

\section{Proof of Corollary~\ref{Corollary2}}
\label{appendix_OP}
The outage probability can be derived from \eqref{eq:snrcdf} in Lemma~\ref{Lemma2} by substituting $\gamma_{th}$ into \eqref{eq:snrcdf}.

\section{Proof of Corollary~\ref{Corollary3}}
\label{appendix_SER}
The conditional SER of the MPSK signal is given as
\begin{equation}
p_{(e/\gamma)}(\gamma) \approx \frac{A}{2} \, \text{erfc} \left( \sin\left( \frac{\pi}{M} \right) \sqrt{\gamma} \right),
\label{eq:ser_conditional}
\end{equation}
where $A=1$ for $M=2$ and $A=2$ for $M>2$.
Since the Rx is in ID mode for $\tau T_{\text{tot}}$ duration, we compute the average SER by taking the integration over the SNR PDF $f_\gamma(\gamma)$ and scaling it by the TSPS time factor $\tau$ as follows:
\begin{equation}
SER_{TSPS} = \tau \int_0^\infty p_{(e|\gamma)}(\gamma) f_\gamma(\gamma) \, d\gamma.
\label{eq:ser_tsps}
\end{equation}
Substituting \eqref{eq:ser_conditional} into \eqref{eq:ser_tsps} and using the PDF of the SNR from Lemma~\ref{Lemma1} and
$\text{erfc}\left( \sqrt{a \gamma} \right) = G^{2,0}_{1,2}\left( a \gamma \left|
\begin{array}{c}
1 \\
0, 0.5
\end{array}
\right.\right), \quad a = \sin^2(\pi/M)
$ gives the final closed-form expression.
This completes the proof.

\section*{Acknowledgement}
This work was supported in part by JST FOREST (Grant JPMJFR2127), in part by JST ASPIRE (Grant JPMJAP2345), in part by JSPS KAKENHI (Grants 23H00470, 24K21615, 24K17272), and in part by the Telecommunications Advancement Foundation.

\bibliographystyle{IEEEtran}
\bibliography{IEEEabrv,ref_macros}

\end{document}